\documentclass[twoside,11pt]{article}

\usepackage[T1]{fontenc}

\usepackage{amsthm}
\usepackage{amsmath}
\usepackage{amssymb}
\usepackage{dsfont}
\usepackage{graphicx}
\usepackage[numbers]{natbib}
\usepackage{nicefrac}
\usepackage{todonotes} 
\setuptodonotes{inline}
\usepackage[UKenglish]{babel}
\usepackage[a4paper,top=2cm,bottom=2cm,left=3cm,right=3cm,marginparwidth=1.75cm]{geometry}
\usepackage{tikz}
\usetikzlibrary{calc}
\usepackage{boxedminipage}
\usepackage{framed}
\usepackage{thm-restate}
\usepackage{tabularx}
\usepackage{tcolorbox}
\usepackage{etoolbox}
\usepackage{xifthen}
\usepackage[colorlinks=true, allcolors=blue]{hyperref}
\usepackage[capitalize,noabbrev]{cleveref}

\newtheorem{theorem}{Theorem}
\newtheorem{proposition}{Proposition}
\newtheorem{lemma}{Lemma}

\newtheorem{claim}{Claim}

\DeclareMathOperator{\dist}{dist}

\newcommand{\ftstc}{\textsc{Fault-Tolerant Path}}
\newcommand{\ftstcshort}{\textsc{FTP}}

\newcommand{\dsl}{\textsc{Directed Steiner Linkage}} 

\newcommand{\SSS}{\ensuremath{|S|}}
\newcommand{\UUU}{\ensuremath{|U|}}
\newcommand{\ALB}{\ensuremath{a}}
\newcommand{\BUB}{\ensuremath{b}}

\newcommand{\Oh}{\ensuremath{\mathcal{O}}}

\newcommand{\noopsort}[2]{#2}

\newlength{\RoundedBoxWidth}
\newsavebox{\GrayRoundedBox}
\newenvironment{GrayBox}[1]%
   {\setlength{\RoundedBoxWidth}{.93\textwidth}
    \def\boxheading{#1}
    \begin{lrbox}{\GrayRoundedBox}
       \begin{minipage}{\RoundedBoxWidth}}%
   {   \end{minipage}
    \end{lrbox}
    \begin{center}
    \begin{tikzpicture}%
       \node(Text)[draw=black!20,fill=white,rounded corners,%
             inner sep=2ex,text width=\RoundedBoxWidth]%
             {\usebox{\GrayRoundedBox}};
        \coordinate(x) at (current bounding box.north west);
        \node [draw=white,rectangle,inner sep=3pt,anchor=north west,fill=white] 
        at ($(x)+(6pt,.75em)$) {\boxheading};
    \end{tikzpicture}
    \end{center}}
    
\newenvironment{defproblemx}[2][]{\noindent\ignorespaces%
                                \FrameSep=6pt%
                                \parindent=0pt%
                \vspace*{-1.5em}
                \ifthenelse{\isempty{#1}}{%
                  \begin{GrayBox}{\textsc{#2}}%
                }{%
                  \begin{GrayBox}{\textsc{#2}  parameterized by~{#1}}%
                }
                \begin{tabular*}{\textwidth}{@{\hspace{.1em}} >{\itshape} p{1.8cm} p{0.8\textwidth} @{}}%
            }{
                \end{tabular*}%
                \end{GrayBox}%
                \ignorespacesafterend
            }  

\newcommand{\defproblem}[3]{
  \begin{defproblemx}{#1}
    Input:  & #2 \\
    Question: & #3
  \end{defproblemx}
}%

\definecolor{darkgreen}{rgb}{0, 0.7, 0}

\begin{document}

\title{When does FTP become FPT?}

\author{Matthias Bentert\thanks{University of Bergen, Norway. \texttt{\{Matthias.Bentert, Fedor.Fomin, Petr.Golovach\}@uib.no}} \and Fedor V. Fomin \addtocounter{footnote}{-1}\footnotemark \and Petr A. Golovach\addtocounter{footnote}{-1}\footnotemark \and 
Laure Morelle\thanks{LIRMM, Université de Montpellier, Montpellier, France. \texttt{laure.morelle@lirmm.fr}}
}
\date{}

\maketitle

\begin{abstract}
In the problem \ftstc{} (\ftstcshort), we are given an edge-weighted directed graph $G = (V, E)$, a subset $U \subseteq E$ of vulnerable edges, two vertices $s, t \in V$, and integers $k$ and $\ell$. The task is to decide whether there exists a subgraph~$H$ of~$G$ with total cost at most~$\ell$ such that, after the removal of any $k$ vulnerable edges, $H$ still contains an $s$-$t$-path.
We study whether \ftstc{} is fixed-parameter tractable (FPT) and whether it admits a polynomial kernel under various parameterizations. Our choices of parameters include:
the number of vulnerable edges in the input graph,  
the number of safe (i.e, invulnerable) edges in the input graph, 
the budget~$\ell$,
the minimum number of safe edges in any optimal solution, 
the minimum number of vulnerable edges in any optimal solution, 
the required redundancy~$k$, and natural above- and below-guarantee parameterizations.
We provide an almost complete description of the complexity landscape of \ftstcshort{} for these parameters.
\end{abstract}

\section{Introduction}
Fault-tolerant edge connectivity refers to the ability of a network to maintain connectivity even after some of its edges fail. Recently, there has been significant interest in a graph model where the edge set is partitioned into two distinct categories: safe edges $S$ and vulnerable edges 
$U$  \cite{AHM22,AHMS22,BentertSS23,BCHI21}. Safe edges represent robust connections that are less likely to fail, while vulnerable edges are more prone to disruptions or failures. Our paper contributes to the expanding body of research on this intriguing model.

Adjiashvili et al.~\cite{AHMS22} introduced the following fault-tolerant generalization of
Minimum-Cost Edge-Disjoint Paths (EDP), 
the classic network connectivity problem whose aim is to find edge-disjoint paths of minimum total cost between two given terminals. 
The input consists of an edge-weighted (directed) graph~${G = (V,E)}$, a subset $U\subseteq E$ of vulnerable edges, two vertices~$s,t \in V$, and integers~$k$ and~$\ell$.
The task is to decide whether there is a subgraph~$H$ of~$G$ of cost at most~$\ell$ such that after removing any~$k$ vulnerable edges from~$H$, there still remains an~$s$-$t$-path. 
More formally, we study the following problem.

\defproblem{\ftstc{} (\ftstcshort)}
{A (directed) graph $G = (V,E=S \cup U)$ where $S$ and $U$ are disjoint sets of \emph{safe} and \emph{vulnerable} edges, respectively, an edge-weight function~$w \colon E \rightarrow \mathds{N}_{\ge 1}$, two vertices~$s,t \in V$, and integers~$k$ and~$\ell$.}
{Is there a set~$K \subseteq E$ of edges of total cost (sum of weights) at most~$\ell$ such that for each set~$F \subseteq K \cap U$ with~$|F| \leq k$, it holds that $G[K \setminus F]$ contains an $s$-$t$-path?}

\vskip-2mm
Note that when $G$ only contains vulnerable edges, i.e., when $E = U$, then \ftstcshort{} is equivalent to deciding whether $G$ contains $k+1$ edge-disjoint $s$-$t$-paths of total cost of at most~$\ell$. In other words, this is equivalent to determining whether a flow of~$k+1$~units can be sent from~$s$ to~$t$ with a cost of at most~$\ell$---assuming all edges in the graph have unit capacity. While the minimum cost capacitated flow is solvable in polynomial time, Adjiashvili et al. showed that \ftstcshort{} is NP-complete~\cite{AHMS22}. They also provided a polynomial-time algorithm for \ftstcshort{} when~${k=1}$. However, the complexity of the problem for any fixed~${k \geq 2}$ remains unknown. Additionally, they presented an $n^{\Oh(k)}$-time algorithm for directed acyclic graphs. Their work is the starting point for our research.

\subsection{Our results}
We investigate the parameterized complexity of \ftstc{} with respect to the following list of natural parameters.
We start with the parameters~$k$ and~$\ell$.
Next, we consider the numbers~$|S|$ and~$|U|$ of safe and vulnerable edges in the input graph as well as the minimum 
numbers~$p$ and~$q$ of safe and vulnerable edges, respectively, in a solution. 
Finally, we also investigate a natural ``above-guarantee'' and a ``below-guarantee'' parameter.
First of all, any solution~$K \subseteq E$ of edges contains an $s$-$t$-path, and thus we have a trivial no-instance if $\ell<\dist(s,t)$, the distance from $s$ to $t$ in the weighted graph.
Hence, we can assume that $\ell\geq\dist(s,t)$, and 
this gives rise to the above-guarantee (that is, above~$\dist(s,t)$) parameter~$\ALB = \ell - \dist(s,t)$. 
To define the last (below-guarantee) parameter, we need the following lemma, which directly follows from \cite[Lemma 3]{AHMS22}. We will use this lemma (sometimes implicitly) in several places. 
 
\begin{lemma}
    \label{lem:sol}
    An instance~$I=(G=(V,S \cup U),s,t,k,\ell)$ of \ftstcshort{} is a yes-instance if and only if there are~$k+1$ $s$-$t$-paths $P_1,P_2,\ldots,P_{k+1}$ with~$P_i = (V_i,E_i)$ in~$G$ such that~$E_i \cap E_j \subseteq S$ for all~$i \neq j \in [k+1]$ and~$ \sum_{e \in \bigcup_{i \in [k+1]} E_i} w(e) \leq \ell$.
\end{lemma}

Given an instance of \ftstc, we define an instance of \textsc{Minimum-Cost Flow} as follows.  
We set the capacity of each safe edge~$e \in S$ to be~$k+1$ and the capacity of each vulnerable edge~$e \in U$ to be one.
Then, we ask for a minimum-cost~$(k+1)$-flow from $s$ to $t$.
We denote the optimum cost of such a flow by~$C$.
We can assume, without loss of generality, that an optimal solution uses integer capacities on all edges and that such a flow solution corresponds to~$k+1$ paths in~$G$ that do not intersect in any vulnerable edges.
By \cref{lem:sol},  there exists a solution of total cost at most~$C$, that is, if~$\ell \geq C$, then the instance~$I$ is a yes-instance.
Note, however, that this is not an if and only if as the flow ``pays'' for each unit of flow that passes through a safe edge, whereas in \ftstc, we only ``pay'' once for each edge.
This leads to the below-guarantee parameter~$\BUB = C - \ell$.
 
\medskip\noindent\textbf{Relations between parameters.}
Before describing our result, let us discuss the relationships between the parameters (see \Cref{fig:results}), most of which are straightforward.
First, the number $p$ of safe edges in the solution is clearly upper-bounded by both the number of safe edges in the input graph and the total cost of edges in the solution (since each edge has a weight of at least one). Hence, if~$p > \ell$, the instance is a trivial no-instance. An analogous statement holds for vulnerable edges.
The remaining parameter bounds are summarized in the following proposition.

\begin{proposition}
    \label{obs:parameters}
If the values of the parameters for an instance of \ftstcshort{} do not satisfy the inequality 
 $k < q \leq 2\ALB \leq 2\ell$, then the considered instance of \ftstcshort{} can be solved in polynomial time.
\end{proposition}

\begin{proof}
     Note that if~$q \leq k$, then after the removal of the vulnerable edges from a solution, we should still have an $s$-$t$-path containing only safe edges.
     Hence, the problem is equivalent to the case where~$q=0$ which can be solved in polynomial time by computing a shortest $s$-$t$-path in the graph induced by all safe edges.

     We next show that~$q \leq 2\ALB$.
     According to \cref{lem:sol}, any optimal solution consists of~$k+1$ paths~$P_1,P_2,\ldots,P_{k+1}$ that pairwise intersect only in safe edges and are of total cost at most~$\ell$.
     If one of these paths only consists of safe edges, then all paths are equal and the solution does not contain any vulnerable edges.
     In this case~$0=q\leq \ALB$.
     Otherwise, note that each vulnerable edge in the solution is contained in exactly one path~$P_i$ and each such path contains at least one vulnerable edge.
     Let~$q_i$ be the number of vulnerable edges in~$P_i$ for each~$i\in [k+1]$.
     Note that~$\sum_{i \in [k+1]}q_i = q$ and~$\sum_{i \in [k]} q_i \leq \ALB$ as path~$P_{k+1}$ contains edges of total cost at least~$\dist(s,t)$ and none of the vulnerable edges in any other path is contained in~$P_{k+1}$.
     Since~$q_{k+1} \leq \ALB$ (for the same reason that~$\sum_{i \in [k]} q_i \leq \ALB$), it follows that~$q = \sum_{i \in [k+1]}q_i \leq 2\ALB$.
     Finally, it holds that~$\ALB = \ell - \dist(s,t) \leq \ell$ as~$\dist(s,t) \geq 0$.
\end{proof}

Our results are depicted in \cref{fig:results}.
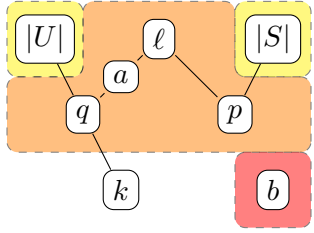
\begin{figure}[t!]
    \centering
    \begin{tikzpicture}
        \fill [opacity=0.5,orange,draw=black,dashed, rounded corners]
        (1.5,.5) -- (1.5,1.5) -- (.5,1.5) -- (.5,2.5) -- (-1.5,2.5) -- (-1.5,1.5) -- (-2.5,1.5) -- (-2.5,.5) -- cycle;

        \fill [opacity=0.5,yellow,draw=black,dashed, rounded corners]
        (1.5,1.5) -- (.5,1.5) -- (.5,2.5) -- (1.5,2.5) -- cycle;
        
        \fill [opacity=0.5,yellow,draw=black,dashed, rounded corners]
        (-2.5,1.5) -- (-1.5,1.5) -- (-1.5,2.5) -- (-2.5,2.5) -- cycle;
        
        \fill [opacity=0.5,red,draw=black,dashed, rounded corners]
        (1.5,-.5) -- (1.5,.5) -- (.5,.5) -- (.5,-.5) -- cycle;

        \node[draw, fill=white, rectangle, rounded corners]  at(-1,0) (k) {$k$};

        \node[draw, fill=white, rectangle, rounded corners]  at(1,0) (b) {$\BUB$};

        \node[draw, fill=white, rectangle, rounded corners]  at(.5,1) (p) {$p$};
        
        \node[draw, fill=white, rectangle, rounded corners]  at(-1.5,1) (q) {$q$} edge(k); 
        
        \node[draw, fill=white, rectangle, rounded corners]  at(-1,1.5) (a) {$\ALB$} edge(q);
        
        \node[draw, fill=white, rectangle, rounded corners] at(1,2) (S) {\SSS} edge(p);

        \node[draw, fill=white, rectangle, rounded corners]  at(-.5,2) (l) {$\ell$} edge(p) edge(a);

        \node[draw, fill=white, rectangle, rounded corners]  at(-2,2) (U) {\UUU} edge(q);
   \end{tikzpicture}
    \caption{Overview of our results. 
    An edge between two parameters~$\alpha$ and~$\beta$, where~$\beta$ is above~$\alpha$, indicates that the value of~$\beta$ in any instance will always be at most a constant times larger than the value of~$\alpha$ in that instance (or the problem can be solved trivially). Any hardness result for~$\beta$ immediately implies the same hardness result for~$\alpha$, and any positive result for~$\alpha$ immediately implies the same positive result for~$\beta$.
    The red box indicates that the parameter leads to para-NP-hardness (NP-hardness for constant parameter values).
    An orange box indicates that the problem is in XP and W[1]-hard and a yellow box indicates that the parameterized problem is fixed-parameter tractable, but does not admit a polynomial kernel.
   }\vskip-3mm
    \label{fig:results}
\end{figure}
As can be seen, we provide an almost complete description of the parameterized complexity landscape of \ftstcshort{} for our selection of natural parameters.
The only remaining question is whether \ftstcshort{} is para-NP-hard or contained in XP when parameterized by~$k$. 
Notice that \ftstcshort{} can be considered for both directed and undirected graphs, and there is no straightforward reduction for   \ftstcshort{} on undirected graphs to the directed case (say, by replacing each undirected edge by a pair of opposite directed ones as for many other problems).
We state all our algorithms for the directed case. However, they can be easily adapted for undirected graphs. Symmetrically, we show our algorithmic lower bounds for unweighted undirected graphs and the results also hold for directed graphs.

\subsection{Related work}
The model of the classic \textsc{$s$-$t$-Path} and \textsc{$s$-$t$-Flow} problems with safe and vulnerable edges is due to Adjiashvili et al.~\cite{AHMS22}, who studied the approximability of these problems. Subsequently, generalizations were also studied \cite{BCGI22,BCHI21,CJ22}.
Closest to our work is the paper by Bentert at al.~\cite{BentertSS23}, who studied the parameterized complexity of computing a fault-tolerant spanning tree.
They provide a fairly complete overview over the parameterized complexity landscape both in terms of natural (problem-specific) as well as structural parameters.

\section{Preliminaries}

In this section, we introduce the notation, definitions, and fundamental concepts used throughout the paper.
For a positive integer $n \in \mathbb{N}$, we denote by $[n]$ the set~$\{1, 2, \dots, n\}$. A \emph{graph} is a pair $G = (V, E)$, where $V$ is a finite set of vertices and $E$ is a set of edges. In an \emph{undirected graph}, edges are unordered pairs of vertices, i.e., ${E \subseteq \{\{u, v\} \mid u, v \in V, u \neq v\}}$. In a \emph{directed graph}, edges are ordered pairs of vertices, i.e., ${E \subseteq \{(u, v) \mid u, v \in V, u \neq v\}}$. Throughout the paper, we use $n$ and $m$ to denote the number of vertices and edges, respectively, if the considered graph is clear from the context. 
The underlying undirected graph of a directed graph~$G = (V,E)$ is the graph where each directed edge is replaced by the corresponding undirected edge (and duplicates are eliminated).
That is, the underlying undirected graph is~${(V,\{\{u,v\} \mid (u,v) \in E \lor (v,u)\in E\})}$. 
A \emph{path} in a graph is a sequence~$(v_1, v_2, \dots, v_k)$ of vertices such that each consecutive pair is connected by an edge, that is, $\{v_i, v_{i+1}\} \in E$ (or~$(v_i,v_{i+1}) \in E$ in the directed case).
A \emph{weighted graph} is a graph~$G = (V, E)$ equipped with a weight function~$w: E \rightarrow \mathbb{Q}_{\geq 0}$ that assigns a non-negative rational to each edge.

\medskip\noindent\textbf{Parameterized Complexity.} Parameterized complexity is a framework for analyzing computational problems based on two inputs: the main input size $n$ and a parameter $k$. A problem is said to be \emph{fixed-parameter tractable} (FPT) if it can be solved in $f(k) \cdot n^{\Oh(1)}$ time, where $f$ is a computable function depending only on $k$. The class $\textsf{FPT}$ contains all such problems.

The class $\textsf{XP}$ consists of problems solvable in time $n^{f(k)}$ for some computable function~$f$. A problem is $\textsf{W[1]}$-\emph{hard} if it is at least as hard as the hardest problems in the class $\textsf{W[1]}$, making it unlikely to be fixed-parameter tractable. A problem is \emph{para-NP-hard} if it remains NP-hard even when the parameter~$k$ is a constant.

A \emph{kernelization} algorithm reduces any instance $(I, k)$ of a parameterized problem to an equivalent instance $(I', k')$ in polynomial time, where the size of~$I'$ and $k'$ depend only on~$k$. If this size is bounded by a polynomial in $k$, the problem admits a \emph{polynomial kernel}.
It is known that a decidable problem is in FPT if and only if it admits a kernel of any size.
Hence, problems admitting a polynomial kernel are a subclass of problems in FPT.
To show that a problem does not admit a polynomial kernel (under standard complexity-theoretic assumptions), one can use polynomial parameter transformations.
A \emph{polynomial parameter transformation} is a polynomial-time reduction from one parameterized problem to another, such that the parameter in the new instance is bounded by a polynomial in the original parameter.
If the original problem does not admit a polynomial parameter, then the problem reduced to also does not have a polynomial kernel.
We refer to the textbook by Cygan et al.~\cite{bluebook} for more information on kernelization and parameterized complexity.

Finally, we introduce three problems that are important later.

\medskip\noindent\textbf{Hitting Set.} Given a finite set $\mathcal{U}$ (the universe), a family $\mathcal{F}$ of subsets of $\mathcal U$, and an integer~$d$, determine whether there exists a subset $H \subseteq \mathcal{U}$ of size at most~$d$ such that $H$ intersects every set in $\mathcal{F}$ (i.e., $H \cap F \neq \emptyset$ for each~$F \in \mathcal{F}$).

\medskip\noindent\textbf{Steiner Tree.} In the \emph{Steiner Tree} problem, we are given an undirected edge-weighted graph $G = (V, E, w)$, a set of terminals $S \subseteq V$, and an integer $d$. The question is whether there is a tree $T = (V_T, E_T)$ in $G$ that spans all vertices in~$S$ (i.e., $S \subseteq V_T$) with total cost at most~$d$, that is, $w(T) = \sum_{e \in E_T} w(e)\le d$.

\medskip\noindent\textbf{Biclique.} A \emph{biclique} in a bipartite graph $G = (U, V, E)$ is a complete bipartite subgraph~$K_{a,b}$ where every vertex in a subset $A \subseteq U$ of size~$a$ is connected to every vertex in a subset $B \subseteq V$ of size~$b$. Given a bipartite graph and an integer~$d$, the \emph{Biclique} problem asks whether~$G$ contains a~$K_{d,d}$.

\section{Algorithmic Results}
In this section, we show our different FPT and XP algorithms.
We mention that all algorithms can handle edge weights and directed graphs.
We start with the two parameters related to the number of vulnerable edges ($q$ and~$|U|$).
Therein, we will use the following problem, which we call \dsl{} and which seems to not have been studied in the literature before.

\defproblem{\dsl}
{A directed graph $G = (V,A)$, an edge-weight function~$w \colon A \rightarrow \mathds{Q}_{\geq 0}$, two multisets~$S$ and~$T$ of vertices with~$|S|=|T|$, and an integer~$\ell$.}
{Is there a set~$K \subseteq A$ of edges of total cost at most~$\ell$ and a bijection~$f \colon S \rightarrow T$ such that for each vertex~$s \in S$, there is a directed path from~$s$ to~$f(s) \in T$ in~$G[K]$?}

We suspect that \dsl{} is of independent interest.
We start by showing that \dsl{} is fixed-parameter tractable when parameterized by the number~$|S|$ of terminals in each of the two terminal sets.
For that algorithm, we first state a helpful structural lemma regarding the structure of an optimal solution.

\begin{lemma}
    \label{lem:dslforest}
    If an instance~$I$ is a yes-instance of \dsl, then there exists a solution whose underlying undirected graph is a forest.
\end{lemma}

\begin{proof}
    Let~$I=(G=(V,E),w,S,T,\ell)$ be a yes-instance of \dsl.
    We represent a solution via a flow formulation.
    Note that \dsl{} is equivalent to asking whether there is a set of edges of total cost at most~$\ell$ such that in the induced graph (with arbitrarily large capacities), there is a feasible flow with one unit of supply in each vertex for each time it appears in~$S$ and one unit of demand in each vertex for each time it appears in~$T$.
    Let~$K$ be a minimal set of edges of total cost at most~$\ell$ such that there exists a feasible flow from~$S$ to~$T$ in~$G[K]$ and let~$C \subseteq K$ be a set of edges in~$K$ whose underlying undirected graph is a cycle.
    Let~$\{v_0,v_1,\ldots,v_{k-1}\}$ be the vertices in~$C$ such that exactly one of~$(v_{i-1},v_{i\bmod k})$ and~$(v_{i\bmod k}, v_{i-1})$ is contained in~$C$ for each~$i \in [k]$.
    Note that if an edge is included in both directions, then reducing the flow along both edges by the same value yields another feasible flow solution.
    Hence, we can reduce the flow along at least one of the two edges to zero and get a solution in which one of the two edges is not used.
    We will next generalize this idea to arbitrary cycles.
    
    Let~$d$ be the minimum flow send over any edge in~$C$ and let~$e_d \in C$ be an edge over which exactly~$d$ flow was sent.
    We assume without loss of generality that~$e_d = (v_{i-1},v_{i \bmod k})$ for some~$i \in [k]$.
    We say that for $j\in[k]$, edges of the form~$(v_{j-1},v_{j \bmod k})$ in~$C$ are \emph{forward} edges and edges of the form~$(v_{j \bmod k},v_{j - 1})$ are \emph{backward} edges.
    We now reduce the flow on each forward edge by~$d$ and increase the flow of each backward edge by~$d$.
    Each vertex incident to two backward edges or two forward edges in~$C$ now has precisely~$d$ less incoming and outgoing flow within~$C$.
    For all other vertices, the total incoming flow and the total outgoing flow remains the same. Thus, for every vertex $v_i$ for $i\in[k]$, the difference between the incoming and outgoing flow remains the same. 
    This means that the change yields another feasible flow solution.  	
    We can now remove all edges with zero flow from the solution.
    As we remove at least the edge~$e_d$, we found a subset of~$K$ which is a solution.
    Since this contradicts the minimality of~$K$, this concludes the proof.
\end{proof}

We next show how to use the previous lemma to show that \dsl{} parameterized by the number of terminals is in FPT.

\begin{theorem}    
    \label{lem:dslfpt}
    \dsl{} can be solved in~$2^{\Oh(k \log k)} n^3$ time, where $k$ is the number of terminals.
\end{theorem}

\begin{proof}
    Let~$I=(G=(V,E),w,S,T,\ell)$ be an instance of \dsl{} and let~${k = |S|}$.
    We first solve all-pairs shortest paths (APSP) in~$G$ in~$\Oh(n^3)$ time using the standard Floyd–Warshall algorithm.
    Then we construct the auxiliary directed weighted graph $H$ with the same set of vertices as $G$. For every pair $(u,v)$ of distinct vertices~${u,v\in V}$, we construct an edge $(u,v)$,
    and set its weight $w'(u,v)=\dist(u,v)$, where~$\dist(u,v)$ is the distance in $G$ between $u$ and $v$, that is, the minimum cost of a directed $u$-$v$-path. We set $w'(u,v)=+\infty$ if such a path does not exist. We show the following claim. 
    
\begin{claim}\label{cl:H}    
 The instances $I=(G,w,S,T,\ell)$ and  $I'=(H,w',S,T,\ell)$ are equivalent. Furthermore, if $(H,w',S,T,\ell)$ is a yes-instance then there is a solution $K'$ such that the underlying graph of $H[K']$ is a forest with at most $2k-2$ internal vertices.   
\end{claim}     
    
 \begin{proof}[Proof of Claim]
 Suppose that $I$ is a yes-instance. By  \cref{lem:dslforest}, we can assume that it has a solution $K$ whose underlying undirected graph is a forest. We assume that $K$ is inclusion minimal. Then $K$ induces a forest $F$ in the underlying graph whose leaves are in~$S\cup T$ (note that some terminals may be internal vertices). Since $|S|+|T|=2k$,
$F$ has at most $2k-2$ internal vertices of degree at least $3$. Let~$X$ be the set of these vertices, and let~$W=X\cup S\cup T$, and let~$\mathcal{P}$ to be the set of paths in $F$ with their end-points in $W$. Notice that each edge of~$K$ is in one of the paths of $\mathcal{P}$.  We  construct a subset $K'$ of the set of edges of $H$ with endpoints in $W$ as follow.  For every path $P=(v_1,\ldots,v_p)\in \mathcal{P}$, we observe that either $P$ or its reversal $(v_p,\ldots,v_1)$ is a directed path in $G$ because the internal vertices are not terminals. We assume without loss of generality that $P$ is a directed path. Then, we add $(v_1,v_p)$ to $K'$. Notice that $w'(v_1,v_p)=\dist(v_1,v_p)$ is at most the cost of~$P$. This immediately implies that the total cost of $K'$ is at most $\ell$. Also, because $K$ is a solution to~$I$, we have that $K'$ is a solution to $I'$, and as $F$ is a forest, we have that underlying graph of $H[K']$ is a forest with at most~$2k-2$ internal vertices. 
 
Assume now that $I'$ is a yes-instance. Then $I'$ has a solution $K'$ of cost at most~$\ell$. For each $(u,v)\in K'$, $G$ contains a $u$-$v$-path $P_{uv}$ of total cost at most~${w'(u,v)=\dist(u,v)}$. Consider the union $K$ of the sets of edges of the paths $P_{uv}$ taken over all edges of $K'$. By construction, the cost of $K$ is at most~$\ell$. Furthermore, because we replace edges by directed paths, $K$ is a solution to~$I$. This proves the claim.
  \end{proof}   
  
 We use this claim and solve  \dsl{} for $I'$. We are looking for a solution whose underlying undirected graph is a forest with at most~$2k-2$ internal vertices.  
  We guess the structure of that forest, that is, we guess a forest $F$ with at most $4k-2$ vertices whose leaves are in $S\cup T$ together with a mapping of its vertices to the terminals. As shown by Tak{\'a}cs~\cite{Tak90}, the number of (undirected) forests with at most~$4k$ labeled vertices can be upper-bounded by~$(4k)^{4k}$.  Thus, we have at most $(4k)^{4k}$ guesses. Next, we turn $F$ to a directed graph $\overrightarrow{F}$ by guessing the orientation of each edge of $F$ in~$\Oh(2^{4k})$ time, and then verify whether this orientation allows for a feasible~$S$-$T$-flow in~$(4k)^{1+o(1)}$ time~\cite{B+23}. Finally, we find a subgraph of $H$ isomorphic to $\overrightarrow{F}$ respecting the mapping of terminals of minimum total cost. Because $\overrightarrow{F}$ has at most $4k-2$ vertices, this can be done by the standard color coding algorithm of Alon, Yuster, and Zwick~\cite{AlonYZ95} in $2^{\Oh(k)}n^2\log n$ time. If the cost of the solution is at most $\ell$, we return the solution. Otherwise, we discard the current choice of~$\overrightarrow{F}$.
  If we fail to find a solution for all possible choices of~$\overrightarrow{F}$, we report that $I'$ is a no-instance. 
  
  Since we have at most $(4k)^{4k}\cdot 2^{4k} \in 2^{\Oh(k \log k)}$ choices for~$\overrightarrow{F}$, the overall running time is in~$2^{\Oh(k \log k)} n^3$. This concludes the proof.
\end{proof}

With \cref{lem:dslfpt} at hand, 
we can show that \ftstcshort{} is FPT parameterized by~$\UUU$ and~XP when parameterized by~$q$.

\begin{theorem}
   \ftstcshort{} can be solved in~$2^{\Oh(\UUU \log(\UUU))} n^3$~time and in~$2^{\Oh(q \log(q))} m^{q}n^3$ time.
\end{theorem}

\begin{proof}
    We first guess the vulnerable edges of an optimal solution.
    This can be done in~$\Oh(2^{\UUU})$ and in~$\Oh(m^q)$~time.
    Next, we reduce the remaining problem to \dsl.
    To this end, we delete all vulnerable edges from the input graph, define~$S$ to be the multiset of the heads of all guessed vulnerable edges and~$k+1$~times the vertex~$s$, and set~$T$ to be the multiset of all tails of guessed vulnerable edges and~$k+1$~times the vertex~$t$.
    By~\cref{obs:parameters}, we may assume~${|S|=|T| \leq q+k+1 \leq 2q \leq 2\UUU}$.
    Next, we use \cref{lem:dslfpt} to solve the constructed instance in~$2^{\Oh(q \log(q))} n^3$~time.
    Note that a solution for \dsl{} (if such a solution exists) together with the guessed set of edges corresponds to a solution for \ftstc.
    This is true since each guessed vulnerable edge in the solution takes one unit of demand out at its tail and adds one unit of supply back at its head.
    Note that this is equivalent to sending one unit of supply through the edge.
    Hence, any solution to \dsl{} together with the guessed vulnerable edges induces a graph in which~$k+1$ units can flow from~$s$ to~$t$ (and there are potentially circles in the graph that circulate some amount of flow).
    Moreover, if the set of vulnerable edges is guessed correctly, then an optimal solution for \dsl{} together with the guessed set of edges also corresponds to an optimal solution for \ftstcshort{} (one in which no flow is circulated).
    To verify this, we show that any solution to \ftstcshort{} without the vulnerable edges is also a solution to \dsl.
    To this end, consider the set of vulnerable edges in an optimal solution.
    \cref{lem:sol} says that the solution can be partitioned into~$k+1$ paths from~$s$ to~$t$ that only intersect in safe edges.
    For each of these paths, consider the set of vulnerable edges in that path in the order they appear along the path.
    Between any two consecutive such edges, the safe edges in the solution connect the head of the first edge to the tail of the second (where we view~$s$ as the head of the first edge and~$t$ as the tail of the last edge).
    This gives both the bijection between the tails and heads of all vulnerable edges in the solution and the solution to connect the respective pairs.
    Thus, it is a solution to \dsl.
    This concludes the proof.
\end{proof}

As we will show later, \ftstcshort{} is W[1]-hard when parameterized by~$\ell$.
As~$\ell$ upper-bounds~$q$, this also shows that we cannot hope to improve the algorithm for~$q$ to an FPT algorithm.
We conclude this section by showing an algorithm that shows fixed-parameter tractability for~$\SSS$ and containment in XP for~$p$.

\begin{theorem}
    \ftstcshort{} can be solved in~$2^{\SSS} \cdot m^{1+o(1)}$~time and in~$m^{p+1+o(1)}$ time.
\end{theorem}

\begin{proof}
    First, we guess the set of safe edges in the solution.
    This can be done in~$\Oh(2^{\SSS})$ and in~$\Oh(m^p)$ time.
    Let~$\ell_1$ be the sum of edge weights of all guessed safe edges.
    Next, we iterate over all edges and assign capacities and costs as follows.
    For safe edges that are guessed to be in the solution, we assign a capacity of~$k+1$ and a cost of~$0$.
    For safe edges that are not guessed to be in the solution, we assign a capacity of~$0$ (which corresponds to removing the edge and hence the cost does not matter).
    For vulnerable edges, we assign a capacity of~$1$ and a cost equal to the weight of the edge.
    Note that the assignment of capacities and costs takes~$\Oh(m)$ time.
    Finally, we compute a minimum-cost $k+1$-flow from~$s$ to~$t$ in the input graph with assigned capacities and costs.
    Note that such a flow corresponds to a solution of \ftstcshort{} by \cref{lem:sol} where we only pay for the vulnerable edges (as all safe edges have a cost of~$0$).
    Hence, we return a solution only if the minimum cost of a~$k+1$ flow plus~$\ell_1$ does not exceed~$\ell$.
    As shown in \cite{B+23}, a minimum-cost flow can be computed in~$m^{1+o(1)}$~time.
    This concludes the proof.
\end{proof}

\section{Hardness Results}
In this section, we complement the results from the previous section.
We start with a 
reduction from \textsc{Biclique} that shows that \ftstcshort{} is W[1]-hard for the parameter~$\ell$ even if~${\BUB = 1}$.
Hence, this also shows para-NP-hardness for~$b$.
Note that the problem is polynomial-time solvable for~$b=0$ and hence our result is tight with respect to~$b$.
All results in this section are for undirected and unweighted graphs, but they also generalize to directed graphs.
\begin{theorem}
\label{prop:W1}
    \ftstcshort{} is W[1]-hard when parameterized by~$\ell$ even when~$\BUB = 1$ and the input graph is undirected and unweighted.
    Moreover, assuming the ETH, the problem cannot be solved in~$f(\ell) \cdot n^{g(\BUB) \cdot o(\sqrt[4]{\ell})}$~time for any computable functions~$f$ and~$g$.
\end{theorem}

\begin{proof}
    We reduce from \textsc{Biclique}.
    This problem is known to be W[1]-hard when parameterized by~$d$ and cannot be solved in~$f(d) \cdot n^{o(\sqrt{d})}$ time for any computable function~$f$ assuming the ETH \cite{Lin18}.
    Let~${(G = (A \cup B,E),d)}$ be an input instance. We assume without loss of generality that $|E|\geq 3d^2$.

    To produce an equivalent instance of \ftstc, we start with a copy of~$G$, where all edges are vulnerable.
    Next, we add three vertices~$s,t,$ and~$v$ and add 
    safe edges between~$v$ and each vertex in~$A$ and between~$t$ and each vertex in~$B$.
    Moreover, we add a path $P$ of~$2d^2-2d-1$ safe edges between~$s$ and~$v$ and~$d^2$ paths between~$s$ and~$t$ of three vulnerable edges each.
    Let the set of edges in these latter paths be~$E_P$.
    Finally, we set~$k=d^2-1$ and~$\ell = 3d^2-1$.
    The reduction is depicted in \cref{fig:param_l}.
    \begin{figure}[t]
        \centering
        \includegraphics[scale=0.8]{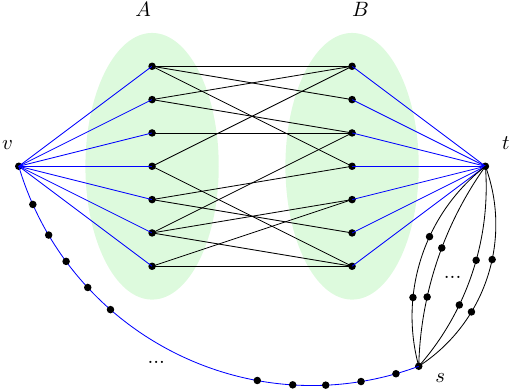}
        \caption{The graph constructed in the proof of \cref{prop:W1}. Safe edges are in blue and vulnerable edges in black.}
        \label{fig:param_l}
        \vskip-3mm
    \end{figure}

    We next show that the two instances are equivalent.
    To this end, assume first that~$G$ contains a~$K_{d,d}$ as a subgraph.
    We build a solution for \ftstcshort{} as follows.
    We pick all edges in the~$s$-$v$-path, all edges in the supposed~$K_{d,d}$, and all safe edges incident to the~$2d$ endpoints.
    Note that this graph contains exactly~${2d^2-2d-1 + d^2 + 2d = \ell}$ edges.
    Moreover, after removing any~${k=d^2-1}$ vulnerable edges, there still remains one edge~$\{u,w\}$ with~$u \in A$ and~$w \in B$ in the solution graph.
    By construction, the~$s$-$v$-path and the edges~$\{v,u\}$ and~$\{w,t\}$ are safe and contained in the solution.
    Thus, there remains an~$s$-$t$-path after removing any~$k$ vulnerable edges.

    For the reverse direction, assume that the constructed instance of \ftstcshort{} is a yes-instance.
    Then, there exists a subgraph consisting of at most~$\ell = 3d^2-1$ edges such that after removing any~$k=d^2-1$ vulnerable edges, there still remains an~$s$-$t$-path.
    Note that this solution contains the entire~$s$-$v$-path $P$  as, otherwise, it has to contain all edges in~$E_P$ and therefore has size at least~$\ell+1$.
    Note that any~$s$-$t$-path in the solution either contains at least one vertex of~$A$ and at least one vertex of~$B$ or is a path with three edges from $E_P$.
    In the second case, we can replace the edges of such a path in the solution by the edges of a path $Q= (v,x,y,t)$, where $x\in A$, $y\in B$, and~$\{x,y\}$ is an edge of $G$ that is not in the solution; such an edge exists because~$|E|\geq 3d^2$. Since~$Q$ contains a single vulnerable edge and the edges of~$P$ are in the original solution, the obtained set of edges forms a solution. Thus, we can assume that any $s$-$t$-path in a solution  contains a vertex from~$A$ and a vertex from~$B$.  
    Denote by $K$ a solution of this type.
    Notice that it should contain at least $k+1=d^2$ vulnerable edges between~$A$ and~$B$. Also, because $\ell=3d^2-1$ and $P$ contains~$2d^2-2d-1$ edges, the total number of edges in $K$ incident to the vertices $A\cup B$ is at most $d^2+2d$.   
     
 Denote by $X\subseteq A$ the set of vertices $x\in A$ such that $\{v,x\}\in K$, and let~$Y\subseteq B$ be the set of vertices $y\in B$ such that $\{y,t\}\in K$. 
 Let also $X'\subseteq A\setminus X$ and~$Y'\subseteq B\setminus Y$ denote the subsets of vertices of $A$ and $B$, respectively, incident to some edges in the solution but not adjacent to $v$ or $t$.   
 Notice that $|X|+|Y|\leq 2d$ because, otherwise, $K$ would contain at most $2d^2-1$ vulnerable edges.  
See \cref{fig:image_sketch} for an illustration.
 \begin{figure}[t]
    \centering
    \includegraphics[width=0.4\linewidth]{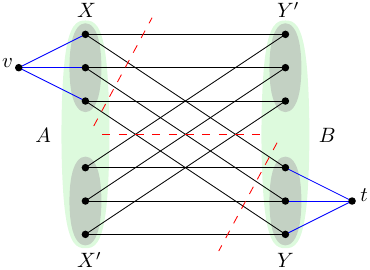}
    \caption{The sets $X,Y,X',$ and $Y'$ constructed in \cref{prop:W1}. The three dashed lines corresponds to $S_{XY}\cup S_{XY'},$ $S_{XY}\cup S_{X'Y},$ and $S_{XY}\cup S_{X'Y'}$.}
    \label{fig:image_sketch}
\end{figure}
Denote by $S_{XY}$ the set of vulnerable edges in $K$ between $X$ and $Y$. Furthermore, 
 denote by $S_{XY'}$, $S_{X'Y}$, and $S_{X'Y'}$ the sets of edges in $K$ between $X$ and $Y'$, $X'$ and $Y$, and $X'$ and $Y'$, respectively. 
 By this definition, $S_{XY}$, $S_{XY'}$, $S_{X'Y}$, and $S_{X'Y'}$ are disjoint, and their union is the set of vulnerable edges in $K$.
 Because $K$ is a solution, we have that 
$|S_{XY}|+|S_{XY'}|\geq d^2$, $|S_{XY}|+|S_{X'Y}|\geq d^2$, and $|S_{XY}|+|S_{X'Y'}|\geq d^2$, as depicted in \cref{fig:image_sketch}. Thus, $K$ contains at least 
$3d^2-2|S_{XY}|$ vulnerable edges. Then 
\begin{equation}\label{eq:SXY}
3d^2-2|S_{XY}|\leq d^2+2d-|X|-|Y|
\end{equation}
 because the total number of edges in $K$ incident to the vertices $A\cup B$ is at most~$d^2+2d$ and $K$ contains $|X|+|Y|$ safe edges. 
Hence, 
\begin{equation*}
2d^2-2d\leq 2|S_{XY}|-|X|-|Y|\leq 2|X||Y|-|X|-|Y|\leq 2\Big(\frac{|X|+|Y|}{2}\Big)^2-2\frac{|X|+|Y|}{2}.
\end{equation*}
This implies that $|X|+|Y|\geq 2d$. We obtain that $|X|+|Y|=2d$ and, moreover, ${|X|=|Y|=d}$ because 
$2|X||Y|-|X|-|Y|=2|X||Y|-2d$ achieves its maximum value for~$|X|=|Y|=d$.  
Then, by Inequality~\ref{eq:SXY}, $|S_{XY}|\geq d^2$. As $|X|+|Y|=2d$, $K$ contains at most $d^2$ vulnerable edges. Thus,
$S_{XY'}=S_{X'Y}=S_{X'Y'}=\emptyset$ and $|S_{XY}|=d^2$. This proves that  
$G[X\cup Y]$ is a biclique $K_{d,d}$. 
This completes the proof of the equivalence of the instances.

    To conclude the proof, observe that the reduction can be computed in polynomial time and that~$\BUB \leq 3d^2 - \ell = 1$ as the~$3d^2$ edges in~$E_P$ ensure that the minimum-cost~$(k+1)$-flow has cost at most~$3d^2$. 
    Thus, any algorithm that solves \ftstcshort{} in~${f(\ell) \cdot n^{g(\BUB) \cdot o(\sqrt[4]{\ell})}}$~time for any computable functions~$f$ and~$g$ can be used to solve \textsc{Biclique} in~${f'(d) \cdot n^{o(\sqrt[4]{d^2})}}$~time for some computable function~$f'$.
    This refutes the ETH~\cite{Lin18}.
\end{proof}

We next exclude polynomial kernels for~$|S|$ and~$|U|$.
We again start with the parameterization by the number of vulnerable edges. The lower bound is obtained by reducing from \textsc{Steiner Tree} parameterized by the number of terminals.

\begin{proposition}
\label{prop:nopolyU}
\ftstcshort{} does not admit a polynomial kernel when parameterized by the number \UUU{} of vulnerable edges in the input graph even if the input graph is undirected and unweighted.
\end{proposition}

\begin{proof}
     We present a polynomial parameter transformation from \textsc{Steiner Tree} parameterized by the number of terminals.
     This problem is known to not admit a polynomial kernel~\cite{bluebook}.
     To this end, let~$(G=(V,E),T \subseteq V,d)$ be an instance of \textsc{Steiner Tree}.
     We start with a copy of~$G$ where all edges are safe edges.
     Next, we add a vertex~$s$ and a vulnerable edge~$\{s,u\}$ for each terminal~$u \in T$.
     Finally, we arbitrarily pick one of the terminals in~$T$ and rename it to~$t$, set~$\ell = d+|T|$, and set~$k = |T|-1$.
     \Cref{fig:param_U} shows an example of the construction.
     \begin{figure}[t]
         \centering
         \includegraphics{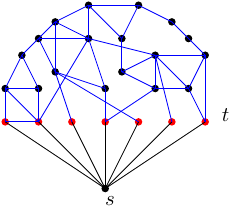}
         \caption{An example of the construction behind the proof of \cref{prop:nopolyU}. The blue (safe) edges are a copy of the input graph~$G$ and the terminals are represented in red. Black edges are vulnerable.}
         \label{fig:param_U}
     \end{figure}

     Since the reduction can clearly be computed in polynomial time and since the number of vulnerable edges is exactly~$|T|$, it only remains to show that the two instances are equivalent.
     To this end, first assume that there exists a Steiner Tree of size at most~$d$ in~$G$.
     Then, we pick a solution by picking the at most~$d$ safe edges in the assumed Steiner Tree plus all~$|T|$ vulnerable edges.
     Note that removing any~$k$ vulnerable edges still leaves a connection between~$s$ and least one terminal vertex.
     Since the~$d$ safe edges ensure that all terminals are connected, we can still find an~$s$-$t$-path in the solution after removing any~$k$ vulnerable edges.

     For the other direction, assume that there is a solution of size~$\ell = d+|T|$ for the constructed instance of \ftstc.
     Note that this solution needs to contain all~$|T|$ vulnerable edges as otherwise we can remove all remaining at most~$k$ vulnerable edges to completely separate~$s$ from the rest of the graph.
     Hence, the solution contains at most~$d$ of the safe edges.
     Now assume that these safe edges do not form a Steiner Tree for~$T$ in~$G$.
     Then, the set of safe edges induce at least two different connected components containing terminal vertices.
     Pick one arbitrary terminal~$w$ that is in a different connected component than~$t$.
     Then, it holds that if we remove all vulnerable edges except for~$\{s,w\}$ from the solution, then~$s$ is only directly connected to~$w$ and by assumption, there is no connection between~$w$ and~$t$.
     This is a contradiction to the fact that we assume that there is always a remaining~$s$-$t$-path after removing any~$k$ vulnerable edges from the solution.
     Thus, the assumption that the set of~$d$ safe edges does not form a Steiner Tree in~$G$ was wrong and this concludes the proof.
\end{proof}

Finally, we exclude polynomial kernels for the parameter~$\SSS$ via a reduction from \textsc{Hitting Set} parameterized by the size of the universe.

\begin{proposition}
\label{prop:nopolyS}
\ftstcshort{} does not admit a polynomial kernel when parameterized by the number \SSS{} of safe edges in the input graph even if the input graph is undirected and unweighted.
\end{proposition}

 \begin{proof}
     We present a polynomial parameter transformation from \textsc{Hitting Set} parameterized by the size of the universe.
     This problem is known to not admit a polynomial kernel~\cite{DLS09}.
     To this end, let~$(\mathcal{U},\mathcal{F},d)$ be an instance of \textsc{Hitting Set}.
     We build a graph as follows.
     We have a vertex~$u_x$ for each element~$x \in \mathcal{U}$ and a vertex~$v_F$ for each set~$F \in \mathcal{F}$.
     We have a vulnerable edge between~$u_x$ and~$v_F$ if and only if~$x \in F$.
     Moreover, we add two vertices~$s$ and~$t$ to~$G$ and add vulnerable edges between~$t$ and each vertex~$v_F$ and safe edges between~$s$ and each vertex~$u_x$.
     Finally, we set~$\ell = 2|\mathcal{F}| + d$ and~$k = |\mathcal{F}|-1$. 
     See \cref{fig:param_S} for an illustration.
\begin{figure}[t]
    \centering
    \includegraphics[scale=0.9]{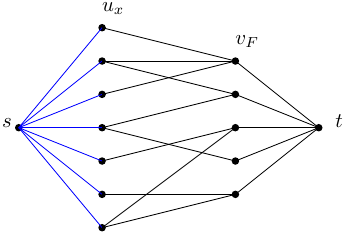}
   \caption{The graph $G$ in the proof of \cref{prop:nopolyS}. Safe edges are represented in blue and vulnerable edges in black.}
    \label{fig:param_S}
\end{figure}

     Since the reduction can clearly be computed in polynomial time and since the number of safe edges is exactly~$|\mathcal{U}|$, it only remains to show that the two instances are equivalent.
     
     To this end, first assume that there exists a hitting set~$H \subseteq \mathcal{U}$ of size at most~$d$.
     Then, we pick a solution by picking the at most~$d$ safe edges in~$\{\{s,u_x\} \mid x \in H\}$ corresponding to~$H$ and for each set~$F \in \mathcal{F}$, we pick the edges~$\{v_F,t\}$ and an arbitrary edge~$\{u_x,v_F\}$ for some~$x \in H \cap F$.
     Note that since~$H$ is a hitting set, such an element always exists.
     Moreover, removing at most~$k$ vulnerable edges from the solution leaves at least one vertex~$v_F$ with both incident vulnerable edges in the remaining solution.
     Since by construction, the two neighbors are~$t$ and a vertex~$u_x$ with~$x \in H$ and the edge~$\{s,u_x\}$ is both safe and contained in the solution, it holds that there is still an~$s$-$t$-path.

     For the other direction, assume that there is a solution of size~${\ell = 2|\mathcal{F}| + d}$ in the constructed instance of \ftstc.
     Then, this solution contains all~${k+1 = |\mathcal{F}|}$ edges incident to~$t$ and at least~$k+1$ vulnerable edges of the form~$\{u_x,v_F\}$ (as otherwise all remaining such edges are a cut between~$s$ and~$t$).
     Hence, the solution contains at most~$d$ safe edges.
     We assume without loss of generality that all vertices~$u_x$ that are connected to~$s$ in the solution are connected to~$s$ via a direct (safe) edge.
     If this is not the case, then there is some vertex~$v_F$ that is connected to~$s$ via a different vertex~$u_{x'}$ and the edge~$\{u_x,v_F\}$ is contained in the solution.
     Note that replacing~$\{u_x,v_F\}$ by~$\{s,u_x\}$ results in another solution.
     Now, let~$H' \subseteq \mathcal{U}$ be the set of elements such that for each~$x \in H'$, the (safe) edge~$\{s,u_x\}$ is contained in the solution.
     By construction, it holds that~$|H'| \leq d$.
     Assume towards a contradiction that~$H'$ is not a hitting set.
     Then, there exists a set~$F' \in \mathcal{F}$ such that~$H' \cap F' = \emptyset$.
     Removing the~$k$ vulnerable edges~$\{v_F,t\}$ for all~$F \neq F'$ disconnects~$s$ from~$t$ as~$t$ is only connected to~$v_{F'}$ and by construction, all neighbors of~$v_{F'}$ in the solution are not contained in~$H'$ and thus not connected to~$s$.
     This concludes the proof.
 \end{proof}

\section{Conclusion}
We analyzed the parameterized complexity of \ftstcshort{} with respect to a collection of natural parameters.
For these parameters, we gave an almost complete classification in terms of polynomial kernels, FPT, XP, and para-NP-hardness.
Our results are summarized in \cref{fig:results}.

The main open question remaining is the exact complexity for the parameter~$k$ (the required redundancy).
Adjiashvili et al.~\cite{AHMS22} showed that the problem is polynomial-time solvable for~$k=1$.
\Cref{prop:W1} shows that the problem is W[1]-hard.
Whether the problem is in XP or para-NP-hard and even whether the problem is polynomial-time solvable for~${k=2}$ remain open.
Adjiashvili et al.~\cite{AHMS22} also showed that \ftstcshort{} is in XP when parameterized by~$k$ if the input graph is a DAG (a directed acyclic graph).
Whether this can be improved to FPT or shown to be W[1]-hard is also an interesting open problem.

Finally, we mention that there is still a lot of room for investigating structural parameters like the treewidth of the input graph.
Moreover, we believe that \dsl{} is a problem of independent interest that deserves further investigation.

\section*{Acknowledgements}
The research leading to these results has received funding from the Research Council of Norway via the project BWCA (grant no. 314528) and the Franco-Norwegian AURORA project (grant no. 349476), 
the European Research Council (ERC) under the European Union’s Horizon 2020 research and innovation programme (grant agreement No. 819416), and the French ANR project ELIT (ANR-20-CE48-0008).

\end{document}